\documentclass{article}
\title{The Generalized Reed-Muller codes in a modular group algebra}
\author{Harinaivo ANDRIATAHINY\\Mention: Mathématiques et informatique,\\Domaine: Sciences et Technologies,\\Université d'Antananarivo, Madagascar\\e-mail: hariandriatahiny@gmail.com\\\\Vololona Harinoro RAKOTOMALALA\\Mention: Météorologie,\\Domaine: Sciences de l'Ingénieur,\\Ecole Supérieure Polytechnique d'Antananarivo,\\ Université d'Antananarivo, Madagascar\\e-mail: volhrakoto@gmail.com}

\usepackage{amsmath,amssymb}
\usepackage{amscd,amsfonts}
\usepackage{amsthm}
\usepackage[english]{babel}
\usepackage[T1]{fontenc}
\usepackage[ansinew]{inputenc}
\swapnumbers
\theoremstyle{plain}
\newtheorem{thm}{Theorem}[section]
\newtheorem{prop}[thm]{Proposition}
\newtheorem{lem}[thm]{Lemma}
\newtheorem{cor}[thm]{Corollary}

\theoremstyle{definition}
\newtheorem{defi}[thm]{Definition}

\newtheorem{rem}[thm]{Remark}

\DeclareMathOperator{\card}{card}

\DeclareMathOperator{\ann}{Ann}
\DeclareMathOperator{\md}{md}
\DeclareMathOperator{\LT}{LT}

\begin{document}

\maketitle

\begin{abstract}
First we study some properties of the modular group algebra $\mathbb{F}_{p^r}[G]$ where $G$ is the additive group of a Galois ring of characteristic $p^r$ and $\mathbb{F}_{p^r}$ is the field of $p^r$ elements. Secondly a description of the Generalized Reed-Muller codes over $\mathbb{F}_{p^r}$ in $\mathbb{F}_{p^r}[G]$ is presented.
\end{abstract}
\section{Introduction}
It is well known that the Reed-Muller codes of length $p^m$ over the prime field $\mathbb{F}_{p}$ are the radical powers of the modular group algebra $\mathbb{F}_{p}[\mathbb{F}_{p^m}]$ ([3],[4]).\\
In general, this property is no longer true for the Generalized Reed-Muller (GRM) codes over a non prime field $\mathbb{F}_{p^{r}} (r>1)$. Let $M$ be the radical of the modular group algebra $\mathbb{F}_{p^r}[G]$ where $G$ is the additive group of the Galois ring $GR(p^r,m)$ of characteristic $p^r$ and of rank $m$. Since $\mathbb{F}_{p^r}[G]$ is isomorphic to $\mathbb{F}_{p^r}[X_{0},X_{1},\ldots ,X_{m-1}]/(X_{0}^{p^r}-1,\ldots ,X_{m-1}^{p^r}-1)$ (see Proposition 5.1), then as shown in [1], except for $M^{0}$, $M$ and $M^{m(p^{r}-1)}$, none of the radical powers of $\mathbb{F}_{p^r}[G]$ is a GRM code over the non prime field $\mathbb{F}_{p^r}$.\\
Many authors have described the GRM codes in group algebras, especially in the group algebra $\mathbb{F}_{p^r}[H]$ where $H$ is an elementary abelian p-group or $H$ is the additive group of an extension field of $\mathbb{F}_{p^r}$ ([4],[5],[9]).\\
We give a description of the GRM codes over an arbitrary finite field $\mathbb{F}_{p^r}$ in the group algebra $\mathbb{F}_{p^r}[G]$ where $G$ is the additive group of the Galois ring $GR(p^r,m)$.\\
We organize this paper as follows. In section 2, we recall the definition and some properties of the Galois ring $GR(p^r,m)$. In section 3, we summarize some results on the group algebra $\mathbb{F}_{p^r}[\mathbb{F}_{p^m}]$ developed by Assmus and Key in [2]. In particular, some facts about a linear basis of the group algebra $\mathbb{F}_{p^r}[\mathbb{F}_{p^m}]$ called Jennings basis ([7]) are pointed out. In section 4, We give some generalizations of the properties of $\mathbb{F}_{p^r}[\mathbb{F}_{p^m}]$ presented in section 3 by considering the group algebra $\mathbb{F}_{p^r}[G]$ where $G$ is the additive group of the Galois ring $GR(p^r,m)$. In section 5, we establish a natural isomorphism between $\mathbb{F}_{p^r}[G]$ and $\mathbb{F}_{p^r}[X_{0},X_{1},\ldots ,X_{m-1}]/(X_{0}^{p^r}-1,\ldots ,X_{m-1}^{p^r}-1)$. In section 6, we describe the GRM codes over a finite field $\mathbb{F}_{p^r}$ in $\mathbb{F}_{p^r}[G]$ by using the fact that GRM codes are extended cyclic codes. In section 7, an example is given.

\section{Galois ring}
Let us start with some basic properties of Galois ring (see [10]).\\
Set $\mathbb{Z}_{p^r}:=\mathbb{Z}/{p^r}\mathbb{Z}$ and let $h(X)\in \mathbb{Z}_{p^r}[X]$ be a monic basic irreducible polynomial of degree $m$. Such a polynomial always exists by Hensel's Lemma. The Galois ring of characteristic $p^r$ and of rank $m$, denoted by $GR(p^r,m)$, is defined as the quotient ring
\begin{equation}\label{galring}
\mathbb{Z}_{p^{r}}[X] / (h(X)).
\end{equation}
In particular, $GR(p^r,1)=\mathbb{Z}_{p^r}$ and $GR(p^1,m)=\mathbb{F}_{p^m}$.\\
$GR(p^r,m)$ is a local ring with maximal ideal $pGR(p^r,m)$ and residue field
\begin{center}
$GR(p^r,m)/pGR(p^r,m) \cong \mathbb{F}_{p^m}$.
\end{center}
The Galois ring $GR(p^r,m)$ is a finite chain ring; all the ideals of $GR(p^r,m)$ are ordered as follows:
\begin{center}
$\left\{0\right\}=p^{r}GR(p^r,m)\subset p^{r-1}GR(p^r,m)\subset \ldots \subset pGR(p^r,m)\subset GR(p^r,m)$.
\end{center}
Each element $v\in GR(p^r,m)$ can be expressed as $v=p^{t}u$ where $u$ is an invertible element in $GR(p^r,m)$ and in this expression $t$ is unique up to modulo $p^{r-t}$.\\
Finally, we recall that if $m_{1}$ is a positive integer such that $m_{1}\mid m$, then $GR(p^r,m_{1})$ is a subring of $GR(p^r,m)$. Conversely, any subring of $GR(p^r,m)$ is of the form $GR(p^r,m_{1})$ for some $m_{1}\mid m$.

\section{The group algebra $\mathbb{F}_{p^r}[\mathbb{F}_{p^m}]$}
We recall in this section some properties of the group algebra $\mathbb{F}_{p^r}[\mathbb{F}_{p^m}]$ treated in [2]. Let $F:=\mathbb{F}_{p^r}$ be any subfield of $\mathbb{F}_{p^m}$. Set $G$ equal to the additive group of the finite field $\mathbb{F}_{p^m}$. Thus $G$ can be viewed as an elementary abelian $p$-group of order $p^{m}$. Set
\begin{center}
$R=F[G]$
\end{center}
the group algebra of $G$ over the field $F$.  Using polynomial notation, an element of $R$, which is a function from $G$ to $F$, is a formal sum $\sum_{g\in G}a_{g}X^{g}$ where the coefficient $a_{g}$ is simply the value of the function at the group element $g$.\\
The operations in the group algebra $R$ are defined as follows:
\begin{equation*}
\sum_{g\in G}a_{g}X^{g}+\sum_{g\in G}b_{g}X^{g}=\sum_{g\in G}(a_{g}+b_{g})X^{g};
\end{equation*}
\begin{equation*}
(\sum_{g\in G}a_{g}X^{g})(\sum_{g\in G}b_{g}X^{g})=\sum_{g,h\in G}a_{g}b_{h}X^{g+h}=\sum_{k\in G}(\sum_{h\in G}a_{k-h}b_{h})X^{k};
\end{equation*}
and, for $c\in F$,
\begin{equation*}
c(\sum_{g\in G}a_{g}X^{g})=\sum_{g\in G}(ca_{g})X^{g}.
\end{equation*}
$X^{0}$ is the multiplicative identity of the group algebra $R$; i.e $X^{0}a=a$ for every $a\in R$. We will simply denote it by $1$.\\
The augmentation map
\begin{equation}
\begin{aligned}
\phi :\quad  R&\longrightarrow F \\
        \sum_{g\in G}a_{g}X^{g} &\longmapsto \sum_{g\in G}a_{g}
\end{aligned}
\end{equation}
is an algebra homomorphism.\\
The kernel of $\phi$ is denoted by
\begin{equation}
M=\left\{\sum_{g\in G}a_{g}X^{g}\in R\mid \sum_{g\in G}a_{g}=0\right\}.
\end{equation}
Thus, $M$ is an ideal of $R$.\\
Since $G$ is an elementary abelian $p$-group, we have
\begin{equation}
(\sum_{g\in G}a_{g}X^{g})^{p}=\sum_{g\in G}(a_{g})^{p}X^{pg}=\sum_{g\in G}(a_{g})^{p}X^{0}=(\sum_{g\in G}a_{g})^{p}X^{0}.
\end{equation}
\begin{prop}
$R$ is a local group algebra with unique maximal ideal\\$M=Rad(R)=\left\{a\in R\mid a^{p}=0\right\}$.
\end{prop}
\begin{rem}
Let $p$ be a prime number. Then
\begin{enumerate}
\item $(-1)^{p-1}=1 \mbox{ mod } p$.
\item $\binom{p-1}{i}=(-1)^{i} \mbox{ mod } p$ for all $i=0,1,\ldots,p-1$.
\end{enumerate}
\end{rem}
\begin{prop}
Provided $g\neq 0$, we have
\begin{equation*}
(X^{g}-1)^{p-1}=\sum_{i=0}^{p-1}X^{ig}
\end{equation*}
\end{prop}
\begin{rem}
Let $U$ be a non-empty subset of $G$. Then
\begin{equation*}
\prod_{g\in U}(X^{g}-1)^{p-1}=\begin{cases}
        \sum_{g\in \left\langle U\right\rangle}X^{g}&\text{if $U$ is a linearly independent set}\\
        0&\text{otherwise}
\end{cases}
\end{equation*}
where $\left\langle U\right\rangle$ is the subspace spanned by $U$.
\end{rem}
Let $S$ be a subset of $R=F[G]$. The annihilator of $S$ is defined by
\begin{center}
$\ann(S):=\left\{a\in R\mid as=0\;\; \text{for all}\;\; s\in S\right\}$.
\end{center}
\begin{thm}
Let $G$ be an elementary abelian $p$-group of order $p^{m}$, $F$ a field of characteristic $p$ and $R=F[G]$.
\begin{enumerate}
\item For any basis $\left\{g_{0},\ldots ,g_{m-1}\right\}$ of $G$, the $p^{m}$ elements $\prod_{\nu =0}^{m-1}(X^{g_{\nu}}-1)^{e_{\nu}}$ where $0\leq e_{\nu}<p$ form a linear basis for $R$.
\item Moreover, $\left\{\prod_{\nu =0}^{m-1}(X^{g_{\nu}}-1)^{e_{\nu}} \mid \sum_{\nu =0}^{m-1}e_{\nu}\geq t , 0\leq e_{\nu}<p\right\}$ form a linear basis of $M^{t}$, where $M$ is the radical of $R$.
\item $\ann(M^{t})=m^{m(p-1)+1-t}$.
\end{enumerate}
\end{thm}
\begin{defi}
The basis for $F[G]$ given in theorem 3.5. is called a Jennings basis of the group algebra.
\end{defi}

\section{The group algebra $\mathbb{F}_{p^r}[GR(p^{r},m)]$}
In this section, we give some natural generalizations of the properties of the group algebra $\mathbb{F}_{p^r}[\mathbb{F}_{p^m}]$ formulated in section 3.\\
Let $F$ be the field $\mathbb{F}_{p^{r}}$ of $p^r$ elements where $p$ is a prime number and $r\geq 1$ is an integer. Set $G$ equal to the additive group of the Galois ring $GR(p^r,m)$. Set
\begin{equation}\label{groupalg}
R=F[G]=\left\{\sum_{g\in G}a_{g}X^{g} \mid a_{g}\in F\right\}
\end{equation}
the group algebra of $G$ over the field $F$. The operations in $R$ are defined as in section 3.\\
We consider the ideal $M:=\left\{\sum_{g\in G}a_{g}X^{g}\in R\mid \sum_{g\in G}a_{g}=0\right\}$ of $R$.\\
Since $G$ is of characteristic $p^{r}$, we have
\begin{equation}
(\sum_{g\in G}a_{g}X^{g})^{p^{r}}=\sum_{g\in G}(a_{g})^{p^{r}}X^{p^{r}g}=\sum_{g\in G}a_{g}X^{0}=(\sum_{g\in G}a_{g})X^{0}.
\end{equation}
\begin{prop}
$R$ is a local group algebra with unique maximal ideal
\begin{center}
$M=Rad(R)=\left\{a\in R\mid a^{p^{r}}=0\right\}$.
\end{center}
\end{prop}
\begin{proof}
The relation (6) shows that $(\sum_{g\in G}a_{g}X^{g})^{p^r}$ is zero if and only if\\$\sum_{g\in G}a_{g}=0$. Thus, $M$ is the radical of $R$.\\
Moreover, it shows also that every element of $R$ that is not in $M$ is a unit of $R$ and it follows that $M$ is the unique maximal ideal of $R$.
\end{proof}
\begin{rem}\label{canbasis}
The set $\left\{X^{g}\mid g\in G\right\}$ is a linear basis of $R$ considered as a vector space over $F$.
\end{rem}
\begin{prop}
The ideal $M$ is generated linearly over the field $F$ by the elements $X^{g}-1$ where $g\in G-\left\{0\right\}$.
\end{prop}
\begin{proof}
Each element $X^{g}-1$ is in $M$. Conversely, let $\sum_{g\in G}a_{g}X^{g}\in M$,  i.e $\sum_{g\in G}a_{g}=0$. We have 
\begin{equation*}
\begin{aligned}
\sum_{g\in G}a_{g}X^{g}&=\sum_{g\in G}a_{g}\left[(X^{g}-1)+1\right]\\&=\sum_{g\in G}\left[a_{g}(X^{g}-1)+a_{g}1\right]=\sum_{g\in G}a_{g}(X^{g}-1)+\sum_{g\in G}a_{g}1\\&=\sum_{g\in G}a_{g}(X^{g}-1)+(\sum_{g\in G}a_{g})1=\sum_{g\in G}a_{g}(X^{g}-1)\\&=\sum_{g\in G-\{0\}}a_{g}(X^{g}-1).
\end{aligned}
\end{equation*}
\end{proof}
\begin{lem}
Let $p$ be a prime number and $r\geq 1$ an integer. Then
\begin{enumerate}
\item $(-1)^{p^{r}-1}=1 \mbox{ mod } p$.
\item $\binom{p^{r}-1}{i}=(-1)^{i} \mbox{ mod } p$ for all $i=0,1,\ldots,p^{r}-1$.
\end{enumerate}
\end{lem}
\begin{proof}
The first assertion is trivial. The second can be proved by induction on $i$.
\end{proof}
\begin{prop}
Provided $g\neq 0$, we have
\begin{equation*}
(X^{g}-1)^{p^{r}-1}=\sum_{i=0}^{p^{r}-1}X^{ig}
\end{equation*}
\end{prop}
\begin{proof}
\begin{equation*}
\begin{aligned}
(X^{g}-1)^{p^{r}-1}&=\sum_{i=0}^{p^{r}-1}\binom{p^{r}-1}{i}(-1)^{p^{r}-1-i}X^{ig}\\&=\sum_{i=0}^{p^{r}-1}\binom{p^{r}-1}{i}(-1)^{p^{r}-1}(-1)^{-i}X^{ig}=\sum_{i=0}^{p^{r}-1}X^{ig}
\end{aligned}
\end{equation*}
by Lemma 4.4.
\end{proof}
From (\ref{galring}), notice that
\begin{equation}\label{gring}
\begin{aligned}
GR(p^r,m)&=\mathbb{Z}_{p^r}[X] / (h(X))\\&=\left\{i_{0}+i_{1}\alpha +\ldots +i_{m-1}\alpha ^{m-1}\mid i_{l}\in \mathbb{Z}_{p^r},0\leq l\leq m-1\right\}
\end{aligned}
\end{equation}
where $\alpha:=X+(h(X))$.
\begin{prop}
We have
\begin{equation*}
(X^{1}-1)^{p^{r}-1}(X^{\alpha}-1)^{p^{r}-1}\ldots (X^{\alpha^{m-1}}-1)^{p^{r}-1}=\sum_{g\in G}X^{g}
\end{equation*}
\end{prop}
\begin{proof}
By Proposition 4.5, we have
\begin{equation*}
\begin{aligned}
&(X^{1}-1)^{p^{r}-1}(X^{\alpha}-1)^{p^{r}-1}\ldots (X^{\alpha^{m-1}}-1)^{p^{r}-1}\\&=(\sum_{i_{0}=0}^{p^{r}-1}X^{i_{0}1})(\sum_{i_{1}=0}^{p^{r}-1}X^{i_{1}\alpha})\ldots (\sum_{i_{m-1}=0}^{p^{r}-1}X^{i_{m-1}\alpha^{m-1}})\\&=\sum_{0\leq i_{l}\leq p^{r}-1,0\leq l\leq m-1}X^{i_{0}1}X^{i_{1}\alpha}\ldots X^{i_{m-1}\alpha^{m-1}}\\&=\sum_{0\leq i_{l}\leq p^{r}-1,0\leq l\leq m-1}X^{i_{0}1+i_{1}\alpha+\ldots +i_{m-1}\alpha^{m-1}}=\sum_{g\in G}X^{g}
\end{aligned}
\end{equation*}
\end{proof}
Set
\begin{equation*}
E:=\left\{0,1,\ldots ,p^{mr}-1\right\}
\end{equation*}
the set of integers between $0$ and $p^{mr}-1$. Let $i$ be an element of $E$ and
\begin{center}
$i=i_{0}+i_{1}p^{r}+i_{2}p^{2r}+\ldots +i_{m-1}p^{(m-1)r}$
\end{center}
where $0\leq i_{l}\leq p^{r}-1$ for all $l=0,1,\ldots ,m-1$ its $p^{r}$-adic expansion.\\
The $p^{r}$-weight of $i$ is defined by
\begin{equation}
\omega_{p^{r}}(i):=\sum_{l=0}^{m-1}i_{l}.
\end{equation}
Then, set
\begin{equation}\label{vectbasis}
V_{i}:=(X^{1}-1)^{i_{0}}(X^{\alpha}-1)^{i_{1}}\ldots (X^{\alpha^{m-1}}-1)^{i_{m-1}}
\end{equation}
and
\begin{equation}\label{radbasis}
B_{t}:=\left\{V_{i}\mid i\in E, \omega_{p^{r}}(i)\geq t\right\}
\end{equation}
for all $t$ such that $0\leq t\leq m(p^r-1)$.
\begin{prop}
$B_{0}$ is a linear basis of $R$ over $F$.
\end{prop}
\begin{proof}
We have $\card(B_{0})=p^{rm}=\dim_{F}(R)$.\\
Let us prove that $B_{0}$ is a linearly independant set over $F$. Consider the linear combination $\sum_{\omega_{p^{r}}(i)\geq 0}\lambda_{i}V_{i}=0$, where $\lambda_{i}\in F$ for all $i$. We have $\lambda_{0}V_{0}+\sum_{\omega_{p^{r}}(i)\geq 1}\lambda_{i}V_{i}=0$ where $V_{0}=1$ is the identity element of $R$. Multiplying by $V_{p^{rm}-1}$, we get $\lambda_{0}V_{p^{rm}-1}+\sum_{\omega_{p^{r}}(i)\geq 1}\lambda_{i}V_{i}V_{p^{rm}-1}=0$.\\
Let $i=\sum_{l=0}^{m-1}i_{l}p^{lr}$ such that $\omega_{p^{r}}(i)\geq 1$. Then, there is an integer $l$ ($0\leq l \leq m-1$) such that $i_{l}\geq 1$. thus, $i_{l}+p^{r}-1\geq p^{r}$. Since $(X^{g}-1)^{p^{r}}=0$ for any $g$, we have $V_{i}V_{p^{rm}-1}=(X^{1}-1)^{i_{0}+p^{r}-1}\ldots (X^{\alpha^{l}}-1)^{i_{l}+p^{r}-1}\ldots (X^{\alpha^{m-1}}-1)^{i_{m-1}+p^{r}-1}=0$. Then, $0=\lambda_{0}V_{p^{rm}-1}=\lambda_{0}\sum_{g\in G}X^{g}$ by Proposition 4.6. Hence, $\lambda_{0}=0$.\\
Thus, we have $\sum_{\omega_{p^{r}}(i)\geq 1}\lambda_{i}V_{i}=0$.\\
Let $i=\sum_{l=0}^{m-1}i_{l}p^{lr}$ such that $\omega_{p^{r}}(i)=1$.\\
We have $\lambda_{i}V_{i}+\sum_{\omega_{p^{r}}(j)\geq 1,j\neq i}\lambda_{j}V_{j}=0$. Multiplying by $V_{p^{rm}-1-i}$, we have $\lambda_{i}V_{i}V_{p^{rm}-1-i}+\sum_{\omega_{p^{r}}(j)\geq 1,j\neq i}\lambda_{j}V_{j}V_{p^{rm}-1-i}=0$.\\
Let $j=\sum_{l=0}^{m-1}j_{l}p^{lr}\neq i$ such that $\omega_{p^{r}}(j)\geq 1=\omega_{p^{r}}(i)$. Then, there is an integer $l$ ($0\leq l\leq m-1$) such that $j_{l}>i_{l}$, and hence $j_{l}\geq i_{l}+1$ and $j_{l}+p^{r}-1-i_{l}\geq (i_{l}+1)+p^{r}-1-i_{l}=p^{r}$. Since $(X^{g}-1)^{p^{r}}=0$ for any $g$, we have $V_{j}V_{p^{rm}-1-i}=(X^{1}-1)^{j_{0}+p^{r}-1-i_{0}}\ldots (X^{\alpha^{l}}-1)^{j_{l}+p^{r}-1-i_{l}}\ldots (X^{\alpha^{m-1}}-1)^{j_{m-1}+p^{r}-1-i_{m-1}}=0$.\\
Consequently, $0=\lambda_{i}V_{i}V_{p^{rm}-1-i}=\lambda_{i}\sum_{g\in G}X^{g}$.\\
Hence $\lambda_{i}=0$.\\
Continuing in this way, we have\\
$\lambda_{i}=0$ for all $i$ such that $\omega_{p^{r}}(i)=1$, and\\
$\lambda_{i}=0$ for all $i$ such that $\omega_{p^{r}}(i)=2,3,\ldots ,m(p^{r}-1)$.
\end{proof}
\begin{cor}
$B_{1}$ is a linear basis of $M$ over $F$.
\end{cor}
\begin{proof}
Since $B_{1}\subset B_{0}$, then $B_{1}$ is a linearly independant set of $R$ over $F$. And since $V_{i}\in M$ for all $i$ such that $\omega_{p^{r}}(i)\geq 1$, $B_{1}\subseteq M$. Moreover, we have $\dim_{F}(M)=p^{rm}-1=\card(B_{1})$.
\end{proof}
\begin{cor}\label{linbasis}
$B_{t}$ is a linear basis of $M^{t}$ over $F$ where $1\leq t\leq m(p^{r}-1)$.
\end{cor}
\begin{proof}
Since $B_{t}\subset B_{0}$, $B_{t}$ is a linearly independant set of $R$ over $F$. $M^{t}$ is generated as ideal by the set $\left\{a_{1}a_{2}\ldots a_{t}\mid a_{i}\in M, 1\leq i\leq t \right\}$. Consider a product $a_{1}a_{2}\ldots a_{t}$ where $a_{i}\in M$ for all $i$ ($1\leq i\leq t$). Since $M$ is linearly generated by $B_{1}$ over $F$, then we have $a_{i}=\sum_{j=1}^{p^{rm}-1}\lambda_{i_{j}}V_{j}$, where $\lambda_{i_{j}}\in F$ for all $i$ and $j$. Therefore, $a_{1}a_{2}\ldots a_{t}=\prod_{i=1}^{t}(\sum_{j=1}^{p^{rm}-1}\lambda_{i_{j}}V_{j})=\sum_{\omega_{p^{r}}(j)\geq t}\beta_{j}V_{j}\in B_{t}$ where $\beta_{j}\in F$. And it is clear that a finite sum of such products is a linear combination of the elements of $B_{t}$ with coefficients in $F$. Thus, $M^{t}$ is generated by $B_{t}$ over $F$.
\end{proof}
\begin{rem}
The index of nilpotency of the radical $M$ is $1+m(p^{r}-1)$. And we have the following sequence of ideals:
\begin{center}
$\left\{0\right\}=M^{1+m(p^{r}-1)}\subset M^{m(p^{r}-1)}\subset\ldots\subset M^{2}\subset M\subset R$.
\end{center}
\end{rem}
\begin{prop}
$\ann(M^{t})= M^{1+m(p^{r}-1)-t}$ where $1\leq t\leq m(p^{r}-1)$.
\end{prop}
\begin{proof}
Let $t$ be an integer such that $1\leq t\leq m(p^{r}-1)$. By Corollary \ref{linbasis}, (\ref{radbasis}), (\ref{vectbasis}) and (7), $B_{t}:=\left\{V_{i}\mid i\in E, \omega_{p^{r}}(i)\geq t\right\}$ is a linear basis of $M^{t}$ over $F$ where $V_{i}:=(X^{1}-1)^{i_{0}}(X^{\alpha}-1)^{i_{1}}\ldots (X^{\alpha^{m-1}}-1)^{i_{m-1}}$ and $i=i_{0}+i_{1}p^{r}+\ldots +i_{m-1}p^{(m-1)r}$ ($0\leq i_{l}\leq p^{r}-1,\forall l$).\\
By definition, $\ann(M^{t}):=\left\{a\in R\mid as=0\;\; \text{for all}\;\; s\in M^{t}\right\}$.\\Since the index of nilpotency of $M$ is $1+m(p^{r}-1)$, then we have $\ann(M^{t})\supseteq M^{1+m(p^{r}-1)-t}$.\\
Let $a\in \ann(M^{t})$. Since $a\in R$, we have $a=\sum_{\omega_{p^{r}}(i)\geq 0}\lambda_{i}V_{i}$ where $\lambda_{i}\in F$ for all $i$.\\
Let $i$ such that $\omega_{p^{r}}(i)=m(p^{r}-1)$. Thus, $i=p^{rm}-1$.\\
Since $V_{p^{rm}-1}\in M^{m(p^{r}-1)}\subseteq M^{t}$, then $aV_{p^{rm}-1}=\lambda_{0}V_{p^{rm}-1}=\lambda_{0}\sum_{g\in G}X^{g}=0$. Hence, $\lambda_{0}=0$.\\
Let $i\in E$ such that $\omega_{p^{r}}(i)=m(p^{r}-1)-1$.\\
Since $V_{i}\in M^{m(p^{r}-1)-1}\subseteq M^{t}$, then $aV_{i}=\lambda_{\bar{i}}\sum_{g\in G}X^{g}=0$ where $\bar{i}=(p^{r}-1-i_{0})+(p^{r}-1-i_{1})p^{r}+\ldots +(p^{r}-1-i_{m-1})p^{(m-1)r}$. Hence $\lambda_{\bar{i}}=0$. We have $\omega_{p^{r}}(\bar{i})=m(p^{r}-1)-\omega_{p^{r}}(i)=1$. Therefore, $\lambda_{\bar{i}}=0$ for all $i$ such that $\omega_{p^{r}}(i)=m(p^{r}-1)-1$.\\
Let $i$ such that $\omega_{p^{r}}(i)=m(p^{r}-1)-2$.\\
Since $V_{i}\in M^{m(p^{r}-1)-2}\subseteq M^{t}$, then $aV_{i}=\lambda_{\bar{i}}\sum_{g\in G}X^{g}=0$. Hence $\lambda_{\bar{i}}=0$. It follows that $\lambda_{\bar{i}}=0$ for all $i$ such that $\omega_{p^{r}}(i)=m(p^{r}-1)-2$, i.e $\omega_{p^{r}}(\bar{i})=2$.
Continuing in this way, finally let $i$ such that $\omega_{p^{r}}(i)=t$.\\
Since $V_{i}\in M^{t}$, then $aV_{i}=\lambda_{\bar{i}}\sum_{g\in G}X^{g}=0$. Hence $\lambda_{\bar{i}}=0$ with $\omega_{p^{r}}(\bar{i})=m(p^{r}-1)-t$.\\
It follows that $\lambda_{\bar{i}}=0$ for all $i$ such that $\omega_{p^{r}}(i)=t$, i.e $\omega_{p^{r}}(\bar{i})=m(p^{r}-1)-t$.\\
Then $a\in M^{m(p^{r}-1)-t+1}$.
\end{proof}

\section{Isomorphism between $\mathbb{F}_{p^r}[GR(p^{r},m)]$ and\\ $\mathbb{F}_{p^r}[X_{0},X_{1},\ldots ,X_{m-1}]/(X_{0}^{p^r}-1,\ldots ,X_{m-1}^{p^r}-1)$}
Recall that $\mathbb{F}_{p^r}[G]=\left\{\sum_{g\in G}a_{g}X^{g}\mid a_{g}\in \mathbb{F}_{p^r}\right\}$ where $G$ is the additive group of the Galois ring $GR(p^r,m)$.
\begin{prop}
The group algebra $\mathbb{F}_{p^r}[G]$ is isomorphic to the quotient ring $\mathbb{F}_{p^{r}}[X_{0},X_{1},\ldots ,X_{m-1}]/(X_{0}^{p^{r}}-1,\ldots ,X_{m-1}^{p^{r}}-1)$.
\end{prop}
\begin{proof}
Consider the mapping
\begin{equation}
\begin{aligned}
\psi :\quad  \mathbb{F}_{p^{r}}[X_{0},X_{1},\ldots ,X_{m-1}]&\longrightarrow R \\
        X_{i} &\longmapsto X^{\alpha^{i}}
\end{aligned}
\end{equation}
where $\alpha$ is defined in (7).\\
It is clear that the ideal $\left\langle X_{i}^{p^{r}}-1\mid 0\leq i\leq m-1\right\rangle \subseteq \ker\psi$.\\
Conversely, let $a(X_{0},\ldots ,X_{m-1})\in \ker\psi$. By using the division algorithm in $\mathbb{F}_{p^{r}}[X_{0},X_{1},\ldots ,X_{m-1}]$ (see[6]) with the lexicographic order such that $X_{0}>X_{1}>\ldots >X_{m-1}$, let us divide $a(X_{0},\ldots ,X_{m-1})$ by\\ $\left\{X_{0}^{p^{r}}-1,\ldots ,X_{m-1}^{p^{r}}-1\right\}$. We have\\
$a(X_{0},\ldots ,X_{m-1})=b_{0}(X_{0}^{p^{r}}-1)+\ldots+b_{m-1}(X_{m-1}^{p^{r}}-1)+g(X_{0},\ldots ,X_{m-1})$ where $b_{i}\in \mathbb{F}_{p^{r}}[X_{0},X_{1},\ldots ,X_{m-1}]$ for all $i$ and\\
$g(X_{0},\ldots ,X_{m-1})=\sum_{0\leq i_{k}\leq p^{r}-1}g_{i_{0},\ldots,i_{m-1}}X_{0}^{i_{0}}\ldots X_{m-1}^{i_{m-1}}$\\
with $g_{i_{0},\ldots,i_{m-1}}\in \mathbb{F}_{p^{r}}$.\\
Then
\begin{equation*}
\begin{aligned}
\psi(a(X_{0},\ldots ,X_{m-1}))&=\psi(g(X_{0},\ldots ,X_{m-1}))\\&=\sum_{0\leq i_{k}\leq p^{r}-1}g_{i_{0},\ldots,i_{m-1}}X^{i_{0}+i_{1}\alpha+\ldots+i_{m-1}\alpha^{m-1}}=0.
\end{aligned}
\end{equation*}
By Remark \ref{canbasis}, $g_{i_{0},\ldots,i_{m-1}}=0$ for all $i_{k}$ such that $0\leq i_{k}\leq p^{r}-1$ ($0\leq k\leq m-1$). This implies that $g(X_{0},\ldots ,X_{m-1})=0$ and thus $a(X_{0},\ldots ,X_{m-1})\in \left\langle X_{i}^{p^{r}}-1\mid 0\leq i\leq m-1\right\rangle$.
\end{proof}

\section{Description of the GRM codes in $\mathbb{F}_{p^r}[GR(p^{r},m)]$}
Set $n:=q^{m}-1$ where $q=p^{r}$, $p$ is a prime number and $r,m$ are integers $\geq 1$. Let $A=\mathbb{F}_{q}[Z]/(Z^{n}-1)$.
We have\\$A=\left\{a_{0}+a_{1}z+\ldots +a_{n-1}z^{n-1}\mid a_{i}\in \mathbb{F}_{q}\right\}$ where $z=Z+(Z^{n}-1)$.\\
The shortened Generalized Reed-Muller code of length $n$ and of order $\nu$ ($0\leq \nu < m(q-1)$) over $\mathbb{F}_{q}$ denoted by $HC_{\nu}(m,q)$ is the ideal of $A$ generated by
\begin{equation}
f_{\nu}(z)=\prod_{0<\omega_{q}(i)\leq m(q-1)-\nu-1}(z-\gamma^{i})
\end{equation}
where $\gamma$ is a primitive element of the finite field $\mathbb{F}_{q^m}$ ([8],[11]), and $\omega_{q}(i)$ is defined in (8).\\
Recall that $G$ is the additive group of the Galois ring $GR(p^r,m)$.
Let us order the elements of $G=\left\{0,g_{0},g_{1},\ldots,g_{q^{m}-2}\right\}$ as follows:
\begin{center}
$0<g_{0}<g_{1}<\ldots<g_{q^{m}-2}$.
\end{center}
Thus, we have the following identification:
\begin{center}
$\sum_{g\in G}a_{g}X^{g}\in \mathbb{F}_{q}[G] \longleftrightarrow (a_{g})_{g\in G} \in (\mathbb{F}_{q})^{q^{m}}$.
\end{center}
Let us order the monomials $\left\{X^{g}\mid g\in G\right\}$ such that $X^{g_{i}}<X^{g_{j}}$ if and only if $g_{i}<g_{j}$.\\
For $a=\sum_{g\in G}a_{g}X^{g}\in \mathbb{F}_{q}[G]$, set $\md(a):=max\left\{g\in G\mid a_{g}\neq 0\right\}$.\\
We define the leading term of $a$ by $\LT(a):=a_{\md(a)}X^{\md(a)}$.\\
Consider the linear transformation
\begin{equation}
\begin{aligned}
\varphi :\quad  A&\longrightarrow \mathbb{F}_{q}[G] \\
        a_{0}+a_{1}z+\ldots +a_{n-1}z^{n-1} &\longmapsto (-\sum_{i=0}^{n-1}a_{i})X^{0}+\sum_{i=0}^{n-1}a_{i}X^{g_{i}}
\end{aligned}
\end{equation}
The Generalized Reed-Muller codes of length $q^{m}$ and of order $\nu$ ($0\leq \nu < m(q-1)$) over $\mathbb{F}_{q}$ is defined by
\begin{center}
$C_{\nu}(m,q)=\varphi(HC_{\nu}(m,q))$.
\end{center}
We have the following sequence:
\begin{equation}\label{seqgrm}
\left\{0\right\}\subset C_{0}(m,q)\subset C_{1}(m,q)\subset\cdots\subset C_{m(q-1)-1}(m,q)\subset \mathbb{F}_{q}[G]
\end{equation}
Since $\varphi$ is a one to one morphism, then we have 
\begin{equation}
\dim_{\mathbb{F}_{q}}(C_{\nu}(m,q))=\dim_{\mathbb{F}_{q}}(HC_{\nu}(m,q))
\end{equation}
with
\begin{equation}
\dim_{\mathbb{F}_{q}}(HC_{\nu}(m,q))=q^{m}-1-\deg(f_{\nu}(z)).
\end{equation}
Set
\begin{equation}\label{extpol}
\theta_{\nu}^{j}=\varphi(z^{j}f_{\nu}(z))
\end{equation}
where $0\leq \nu\leq m(q-1)-1$.\\
By (13), we have
\begin{equation}
\LT(\theta_{\nu}^{j})=X^{g_{\deg(z^{j}f_{\nu}(z))}}.
\end{equation}
Set $N(t):=\card(\left\{k\in\left[0,q^{m}-1\right]\mid \omega_{q}(k)=t\right\})$ where $\left[0,q^{m}-1\right]$ is the set of integers between $0$ and $q^{m}-1$.
\begin{thm}\label{propal}
We have
\begin{enumerate}
\item $C_{0}(m,q)$ is linearly generated over $\mathbb{F}_{q}$ by $\theta_{0}^{0}$ .
\item $C_{\nu}(m,q)$ is linearly generated over $\mathbb{F}_{q}$ by the set\\ $K_{\nu}:=\left\{\theta_{0}^{0}\right\}\cup\left\{\theta_{i}^{j}\mid 1\leq i\leq \nu \;,\; 0\leq j\leq N(m(q-1)-i)-1\right\}$ where $1\leq \nu\leq m(q-1)-1$.
\end{enumerate}
\end{thm}
\begin{proof}
Let $\nu$ be an integer such that $0\leq \nu\leq m(q-1)-1$.\\ We have $\theta_{0}^{0}=\varphi(f_{0}(z))=\varphi(\prod_{i=1}^{q^{m}-2}(z-\gamma^{i}))=\varphi(1+z+z^{2}+\ldots +z^{q^{m}-2})=\sum_{g\in G}X^{g}$.\\
Since $\dim_{\mathbb{F}_{q}}(C_{0}(m,q))=\dim_{\mathbb{F}_{q}}(HC_{0}(m,q))=q^{m}-1-(q^{m}-2)=1$, then the first assertion is proved.\\
For the integer $\nu$ such that $1\leq \nu \leq m(q-1)-1$, $f_{\nu}(z)$ is a monic polynomial and by (12), $\deg(f_{\nu}(z))=\sum_{t=1}^{m(q-1)-\nu-1}N(t)$.\\
Thus, $\deg(z^{j}f_{\nu}(z))=j+\sum_{t=1}^{m(q-1)-\nu-1}N(t)$ for $1\leq j\leq N(m(q-1)-\nu)-1$.\\
By (18), for the elements $\theta_{i}^{j}$ in $K_{\nu}$, we have
\begin{center}
$\LT(\theta_{\lambda}^{j})<\LT(\theta_{\mu}^{k}) \;\text{if}\; \lambda > \mu$, and $\LT(\theta_{\mu}^{j})<\LT(\theta_{\mu}^{k})\; \text{if}\; j<k$.
\end{center}
This implies that the elements of $K_{\nu}$ are linearly independant. And by (14), they are elements of $C_{\nu}(m,q)$.\\
It is clear that $\card(K_{\nu})=1+\sum_{i=1}^{\nu}N(m(q-1)-i)$.\\
And since
\begin{equation*}
\begin{aligned}
N(m(q-1)-i)&=\deg(f_{i-1}(z))-\deg(f_{i}(z))\\&=(q^{m}-1-\deg(f_{i}(z)))-(q^{m}-1-\deg(f_{i-1}(z)))\\&=
\dim_{\mathbb{F}_{q}}(HC_{i}(m,q))-\dim_{\mathbb{F}_{q}}(HC_{i-1}(m,q))\\&=\dim_{\mathbb{F}_{q}}(C_{i}(m,q))-\dim_{\mathbb{F}_{q}}(C_{i-1}(m,q)),
\end{aligned}
\end{equation*}
then $\card(K_{\nu})=\dim_{\mathbb{F}_{q}}(C_{\nu}(m,q))$.
\end{proof}
For $m=1$, we consider the case of the Reed-Solomon codes $C_{\nu}(1,q)$ of length $q$ and of order $\nu$ over $\mathbb{F}_{q}$.
\begin{cor}
$C_{\nu}(1,q)$ is linearly generated over $\mathbb{F}_{q}$ by the set\\ $\left\{\theta_{i}^{0}\mid 0\leq i\leq \nu \right\}$ where $0\leq \nu\leq q-2$.
\end{cor}

\section{Example: the GRM codes of lenghth $16$ over $\mathbb{F}_{4}$ in $\mathbb{F}_{4}[GR(2^{2},2)]$}
Let $G=\left\{0,g_{0},g_{1},\ldots,g_{14}\right\}$ be the additive group of the Galois ring $GR(2^{2},2)$ with $0<g_{0}<g_{1}<\ldots<g_{14}$.\\
\underline{$C_{0}(2,4)$:}\\
We have $C_{0}(2,4)=\varphi(HC_{0}(2,4))$ with $HC_{0}(2,4)=(f_{0}(z))$ and 
\begin{equation*}
f_{0}(z)=\prod_{0<\omega_{4}(i)\leq 5}(z-\gamma^{i})=\prod_{i=1}^{14}(z-\gamma^{i})=1+z+z^{2}+\ldots +z^{14}
\end{equation*}
where $\gamma$ is a primitive element of $\mathbb{F}_{16}$.
We have $\theta_{0}^{0}=\varphi(f_{0}(z))=X^{0}+X^{g_{0}}+X^{g_{1}}+\ldots +X^{g_{14}}$. Thus, $\LT(\theta_{0}^{0})=X^{g_{14}}$.\\
And by (15) and (16), we have\\ $\dim_{\mathbb{F}_{4}}(C_{0}(2,4))=\dim_{\mathbb{F}_{4}}(HC_{0}(2,4))=15-\deg(f_{0}(z))=15-14=1$.\\
Thus, $C_{0}(2,4)$ is linearly generated over $\mathbb{F}_{4}$ by $\theta_{0}^{0}$.\\
\underline{$C_{1}(2,4)$:}\\
We have $C_{1}(2,4)=\varphi(HC_{1}(2,4))$ with $HC_{1}(2,4)=(f_{1}(z))$ and 
\begin{equation*}
f_{1}(z)=\prod_{0<\omega_{4}(i)\leq 4}(z-\gamma^{i}).
\end{equation*}
Consider $N(5):=\card(\left\{k\in\left[0,15\right]\mid \omega_{4}(k)=5\right\})=\card(\left\{11,14\right\})=2$.\\
We have $\deg(f_{1}(z))=\deg(f_{0}(z))-N(5)=14-2=12$.\\
And by (15) and (16),\\ $\dim_{\mathbb{F}_{4}}(C_{1}(2,4))=\dim_{\mathbb{F}_{4}}(HC_{1}(2,4))=15-\deg(f_{1}(z))=15-12=3$.
On the other hand, by (18),\\
since $\theta_{1}^{0}=\varphi(f_{1}(z))$, then $\LT(\theta_{1}^{0})=X^{g_{12}}$,\\
and since $\theta_{1}^{1}=\varphi(zf_{1}(z))$, then $\LT(\theta_{1}^{1})=X^{g_{13}}$.\\
By (14), we have $C_{0}(2,4)\subset C_{1}(2,4)$. Thus, $\left\{\theta_{0}^{0},\theta_{1}^{0},\theta_{1}^{1}\right\}=K_{1}$ is a linearly independent set in $C_{1}(2,4)$. Then it is a linear basis of $C_{1}(2,4)$.\\
\underline{$C_{2}(2,4)$:}\\
We have $C_{2}(2,4)=\varphi(HC_{2}(2,4))$ with $HC_{2}(2,4)=(f_{2}(z))$ and 
\begin{equation*}
f_{2}(z)=\prod_{0<\omega_{4}(i)\leq 3}(z-\gamma^{i}).
\end{equation*}
Consider $N(4):=\card(\left\{k\in\left[0,15\right]\mid \omega_{4}(k)=4\right\})=\card(\left\{7,10,13\right\})=3$.\\
We have $\deg(f_{2}(z))=\deg(f_{1}(z))-N(4)=12-3=9$.\\
And by (15) and (16),\\ $\dim_{\mathbb{F}_{4}}(C_{2}(2,4))=\dim_{\mathbb{F}_{4}}(HC_{2}(2,4))=15-\deg(f_{2}(z))=15-9=6$.
On the other hand, by (18),\\
- since $\theta_{2}^{0}=\varphi(f_{2}(z))$, then $\LT(\theta_{2}^{0})=X^{g_{9}}$,\\
- since $\theta_{2}^{1}=\varphi(zf_{2}(z))$, then $\LT(\theta_{2}^{1})=X^{g_{10}}$,\\
- since $\theta_{2}^{2}=\varphi(z^{2}f_{2}(z))$, then $\LT(\theta_{2}^{2})=X^{g_{11}}$.\\
By (14), we have $C_{0}(2,4)\subset C_{1}(2,4)\subset C_{2}(2,4)$. Thus, $\left\{\theta_{0}^{0},\theta_{1}^{0},\theta_{1}^{1},\theta_{2}^{0},\theta_{2}^{1},\theta_{2}^{2}\right\}=K_{2}$ is a linearly independent set in $C_{2}(2,4)$. Then it is a linear basis of $C_{2}(2,4)$.\\
\underline{$C_{3}(2,4)$:}\\
We have $C_{3}(2,4)=\varphi(HC_{3}(2,4))$ with $HC_{3}(2,4)=(f_{3}(z))$ and 
\begin{equation*}
f_{3}(z)=\prod_{0<\omega_{4}(i)\leq 2}(z-\gamma^{i}).
\end{equation*}
Consider $N(3):=\card(\left\{k\in\left[0,15\right]\mid \omega_{4}(k)=3\right\})=\card(\left\{3,6,9,12\right\})=4$.\\
We have $\deg(f_{3}(z))=\deg(f_{2}(z))-N(3)=9-4=5$.\\
And by (15) and (16),\\ $\dim_{\mathbb{F}_{4}}(C_{3}(2,4))=\dim_{\mathbb{F}_{4}}(HC_{3}(2,4))=15-\deg(f_{3}(z))=15-5=10$.
On the other hand, by (18),\\
- since $\theta_{3}^{0}=\varphi(f_{3}(z))$, then $\LT(\theta_{3}^{0})=X^{g_{5}}$,\\
- since $\theta_{3}^{1}=\varphi(zf_{3}(z))=6$, then $\LT(\theta_{3}^{1})=X^{g_{6}}$,\\
- since $\theta_{3}^{2}=\varphi(z^{2}f_{3}(z))=7$, then $\LT(\theta_{3}^{2})=X^{g_{7}}$,\\
and since $\theta_{3}^{3}=\varphi(z^{3}f_{3}(z))=8$, then $\LT(\theta_{3}^{3})=X^{g_{8}}$.\\
By (14), we have $C_{0}(2,4)\subset C_{1}(2,4)\subset C_{2}(2,4)\subset C_{3}(2,4)$. Thus,\\ $\left\{\theta_{0}^{0},\theta_{1}^{0},\theta_{1}^{1},\theta_{2}^{0},\theta_{2}^{1},\theta_{2}^{2},\theta_{3}^{0},\theta_{3}^{1},\theta_{3}^{2},\theta_{3}^{3}\right\}=K_{3}$ is a linearly independent set in $C_{3}(2,4)$. Then it is a linear basis of $C_{3}(2,4)$.\\
\underline{$C_{4}(2,4):$}\\
We have $C_{4}(2,4)=\varphi(HC_{4}(2,4))$ with $HC_{4}(2,4)=(f_{4}(z))$ and 
\begin{equation*}
f_{4}(z)=\prod_{0<\omega_{4}(i)\leq 1}(z-\gamma^{i}).
\end{equation*}
Consider $N(2):=\card(\left\{k\in\left[0,15\right]\mid \omega_{4}(k)=2\right\})=\card(\left\{2,5,8\right\})=3$.\\
We have $\deg(f_{4}(z))=\deg(f_{3}(z))-N(2)=5-3=2$.\\
And by (15) and (16),\\ $\dim_{\mathbb{F}_{4}}(C_{4}(2,4))=\dim_{\mathbb{F}_{4}}(HC_{4}(2,4))=15-\deg(f_{4}(z))=15-2=13$.
On the other hand, by (18),\\
- since $\theta_{4}^{0}=\varphi(f_{4}(z))$, then $\LT(\theta_{4}^{0})=X^{g_{2}}$,\\
- since $\theta_{4}^{1}=\varphi(zf_{4}(z))=3$, then $\LT(\theta_{4}^{1})=X^{g_{3}}$,\\
- since $\theta_{4}^{2}=\varphi(z^{2}f_{4}(z))=4$, then $\LT(\theta_{4}^{2})=X^{g_{4}}$.\\
By (14), we have $C_{0}(2,4)\subset C_{1}(2,4)\subset C_{2}(2,4)\subset C_{3}(2,4)\subset C_{4}(2,4)$. Thus,\\ $\left\{\theta_{0}^{0},\theta_{1}^{0},\theta_{1}^{1},\theta_{2}^{0},\theta_{2}^{1},\theta_{2}^{2},\theta_{3}^{0},\theta_{3}^{1},\theta_{3}^{2},\theta_{3}^{3},\theta_{4}^{0},\theta_{4}^{1},\theta_{4}^{2}\right\}=K_{4}$ is a linearly independent set in $C_{4}(2,4)$. Then it is a linear basis of $C_{4}(2,4)$.\\
\underline{$C_{5}(2,4)$:}\\
We have $C_{5}(2,4)=\varphi(HC_{5}(2,4))$ with $HC_{5}(2,4)=(f_{5}(z))$ and 
\begin{equation*}
f_{5}(z)=\prod_{0<\omega_{4}(i)\leq 0}(z-\gamma^{i})=1.
\end{equation*}
Then $\deg(f_{5}(z))=0$.\\
And by (15) and (16),\\ $\dim_{\mathbb{F}_{4}}(C_{5}(2,4))=\dim_{\mathbb{F}_{4}}(HC_{5}(2,4))=15-\deg(f_{5}(z))=15$.
On the other hand, by (18),\\
- since $\theta_{5}^{0}=\varphi(f_{5}(z))$, then $\LT(\theta_{5}^{0})=X^{g_{0}}$,\\
- since $\theta_{5}^{1}=\varphi(zf_{5}(z))=1$, then $\LT(\theta_{5}^{1})=X^{g_{1}}$.\\
By (14), we have $C_{0}(2,4)\subset C_{1}(2,4)\subset C_{2}(2,4)\subset C_{3}(2,4)\subset C_{4}(2,4)\subset C_{5}(2,4)$. Thus,\\ $\left\{\theta_{0}^{0},\theta_{1}^{0},\theta_{1}^{1},\theta_{2}^{0},\theta_{2}^{1},\theta_{2}^{2},\theta_{3}^{0},\theta_{3}^{1},\theta_{3}^{2},\theta_{3}^{3},\theta_{4}^{0},\theta_{4}^{1},\theta_{4}^{2},\theta_{5}^{0},\theta_{5}^{1}\right\}=K_{5}$ is a linearly independent set in $C_{5}(2,4)$. Then it is a linear basis of $C_{5}(2,4)$.

\end{document}